\documentclass{article}

\def\noheaderplainsetup{

\topmargin=0pt \headheight=0pt \headsep=0pt  \oddsidemargin=0pt \evensidemargin=0pt  \textheight=9.0truein \textwidth=6.2truein}

\noheaderplainsetup

\usepackage{amsfonts}
\usepackage{cite}

\begin{document}

%     MISC.:

\newcommand{\code}[1]{\ulcorner #1 \urcorner}
\newcommand{\mldi}{\hspace{2pt}\mbox{\footnotesize $\vee$}\hspace{2pt}}
\newcommand{\mlci}{\hspace{2pt}\mbox{\footnotesize $\wedge$}\hspace{2pt}}
\newcommand{\emptyrun}{\langle\rangle} %empty run
\newcommand{\oo}{\bot}            %opponent, the corresponding truth value and label
\newcommand{\pp}{\top}            %proponent, the corresponding truth value and label
\newcommand{\xx}{\wp}               %player, the corresponding truth value and label
\newcommand{\legal}[2]{\mbox{\bf Lr}^{#1}_{#2}} %function telling what the legal positions are
\newcommand{\win}[2]{\mbox{\bf Wn}^{#1}_{#2}} %function telling who the winner is
 \newcommand{\one}{\mbox{\sc One}}
 \newcommand{\two}{\mbox{\sc Two}}
 \newcommand{\three}{\mbox{\sc Three}}
 \newcommand{\four}{\mbox{\sc Four}}
 \newcommand{\first}{\mbox{\sc Derivation}}
 \newcommand{\second}{\mbox{\sc Second}}
 \newcommand{\uorigin}{\mbox{\sc Org}}
 \newcommand{\image}{\mbox{\sc Img}}
 \newcommand{\limitset}{\mbox{\sc Lim}}
 \newcommand{\fif}{\mbox{\bf CL15}}
\newcommand{\col}[1]{\mbox{$#1$:}}

\newcommand{\sti}{\mbox{\raisebox{-0.02cm}
{\scriptsize $\circ$}\hspace{-0.121cm}\raisebox{0.08cm}{\tiny $.$}\hspace{-0.079cm}\raisebox{0.10cm}
{\tiny $.$}\hspace{-0.079cm}\raisebox{0.12cm}{\tiny $.$}\hspace{-0.085cm}\raisebox{0.14cm}
{\tiny $.$}\hspace{-0.079cm}\raisebox{0.16cm}{\tiny $.$}\hspace{1pt}}}
\newcommand{\costi}{\mbox{\raisebox{0.08cm}
{\scriptsize $\circ$}\hspace{-0.121cm}\raisebox{-0.01cm}{\tiny $.$}\hspace{-0.079cm}\raisebox{0.01cm}
{\tiny $.$}\hspace{-0.079cm}\raisebox{0.03cm}{\tiny $.$}\hspace{-0.085cm}\raisebox{0.05cm}
{\tiny $.$}\hspace{-0.079cm}\raisebox{0.07cm}{\tiny $.$}\hspace{1pt}}}

\newcommand{\seq}[1]{\langle #1 \rangle}           % sequence: <...>

\newcommand{\pstb}{\mbox{\raisebox{-0.01cm}{\large $\wedge$}\hspace{-5pt}\raisebox{0.26cm}{\small $\mid$}\hspace{4pt}}}
\newcommand{\pcostb}{\mbox{\raisebox{0.22cm}{\large $\vee$}\hspace{-5pt}\raisebox{0.02cm}{\footnotesize $\mid$}\hspace{4pt}}}

\newcommand{\sst}{\mbox{\raisebox{-0.07cm}{\scriptsize $-$}\hspace{-0.2cm}$\pst$}} % sequential recurrence

\newcommand{\scost}{\mbox{\raisebox{0.20cm}{\scriptsize $-$}\hspace{-0.2cm}$\pcost$}} % sequential corecurrence

%     OPERATORS:

\newcommand{\mla}{\mbox{{\Large $\wedge$}}}
\newcommand{\mle}{\mbox{{\Large $\vee$}}}

\newcommand{\pst}{\mbox{\raisebox{-0.01cm}{\scriptsize $\wedge$}\hspace{-4pt}\raisebox{0.16cm}{\tiny $\mid$}\hspace{2pt}}}
\newcommand{\gneg}{\neg}                  %game negation
\newcommand{\mli}{\rightarrow}                     %strong reduction
\newcommand{\cla}{\mbox{\large $\forall$}}      %blind universal quantifier
\newcommand{\cle}{\mbox{\large $\exists$}}        %blind existential quantifier
\newcommand{\mld}{\vee}    %multiplicative disjunction
\newcommand{\mlc}{\wedge}  %multiplicative conjunction
\newcommand{\ade}{\mbox{\Large $\sqcup$}}      %additive existential quantifier
\newcommand{\ada}{\mbox{\Large $\sqcap$}}      %additive universal quantifier
\newcommand{\add}{\sqcup}                      %additive disjunction
\newcommand{\adc}{\sqcap}                      %additive conjunction

\newcommand{\tlg}{\bot}               %classical \bot; trivially lost elementary game
\newcommand{\twg}{\top}               %classical \top; trivially won elementary game
\newcommand{\st}{\mbox{\raisebox{-0.05cm}{$\circ$}\hspace{-0.13cm}\raisebox{0.16cm}{\tiny $\mid$}\hspace{2pt}}}
\newcommand{\cst}{{\mbox{\raisebox{-0.05cm}{$\circ$}\hspace{-0.13cm}\raisebox{0.16cm}{\tiny $\mid$}\hspace{1pt}}}^{\aleph_0}} % countable recurrence
\newcommand{\cost}{\mbox{\raisebox{0.12cm}{$\circ$}\hspace{-0.13cm}\raisebox{0.02cm}{\tiny $\mid$}\hspace{2pt}}}
\newcommand{\ccost}{{\mbox{\raisebox{0.12cm}{$\circ$}\hspace{-0.13cm}\raisebox{0.02cm}{\tiny $\mid$}\hspace{1pt}}}^{\aleph_0}} % countable corecurrence
\newcommand{\pcost}{\mbox{\raisebox{0.12cm}{\scriptsize $\vee$}\hspace{-4pt}\raisebox{0.02cm}{\tiny $\mid$}\hspace{2pt}}}

\newcommand{\psti}{\mbox{\raisebox{-0.02cm}{\tiny $\wedge$}\hspace{-0.121cm}\raisebox{0.08cm}{\tiny $.$}\hspace{-0.079cm}\raisebox{0.10cm}
{\tiny $.$}\hspace{-0.079cm}\raisebox{0.12cm}{\tiny $.$}\hspace{-0.085cm}\raisebox{0.14cm}
{\tiny $.$}\hspace{-0.079cm}\raisebox{0.16cm}{\tiny $.$}\hspace{1pt}}}

\newcommand{\pcosti}{\mbox{\raisebox{0.08cm}{\tiny $\vee$}\hspace{-0.121cm}\raisebox{-0.01cm}{\tiny $.$}\hspace{-0.079cm}\raisebox{0.01cm}
{\tiny $.$}\hspace{-0.079cm}\raisebox{0.03cm}{\tiny $.$}\hspace{-0.085cm}\raisebox{0.05cm}
{\tiny $.$}\hspace{-0.079cm}\raisebox{0.07cm}{\tiny $.$}\hspace{1pt}}}

%   NUMERATED ITEMS and ENVIRONMENTS

\newtheorem{theoremm}{Theorem}[section]
\newtheorem{conditionss}{Condition}[section]
\newtheorem{thesiss}[theoremm]{Thesis}
\newtheorem{definitionn}[theoremm]{Definition}
\newtheorem{lemmaa}[theoremm]{Lemma}
\newtheorem{notationn}[theoremm]{Notation}\newtheorem{corollary}[theoremm]{Corollary}
\newtheorem{propositionn}[theoremm]{Proposition}
\newtheorem{conventionn}[theoremm]{Convention}
\newtheorem{examplee}[theoremm]{Example}
\newtheorem{remarkk}[theoremm]{Remark}
\newtheorem{factt}[theoremm]{Fact}
\newtheorem{exercisee}[theoremm]{Exercise}
\newtheorem{questionn}[theoremm]{Open Problem}
\newtheorem{conjecturee}[theoremm]{Conjecture}

\newenvironment{exercise}{\begin{exercisee} \em}{ \end{exercisee}}
\newenvironment{definition}{\begin{definitionn} \em}{ \end{definitionn}}
\newenvironment{theorem}{\begin{theoremm}}{\end{theoremm}}
\newenvironment{lemma}{\begin{lemmaa}}{\end{lemmaa}}
\newenvironment{proposition}{\begin{propositionn} }{\end{propositionn}}
\newenvironment{convention}{\begin{conventionn} \em}{\end{conventionn}}
\newenvironment{remark}{\begin{remarkk} \em}{\end{remarkk}}
\newenvironment{proof}{ {\bf Proof.} }{\  \rule{2.5mm}{2.5mm} \vspace{.2in} }
\newenvironment{idea}{ {\bf Idea.} }{\  \rule{1.5mm}{1.5mm} \vspace{.15in} }
\newenvironment{example}{\begin{examplee} \em}{\end{examplee}}
\newenvironment{fact}{\begin{factt}}{\end{factt}}
\newenvironment{notation}{\begin{notationn} \em}{\end{notationn}}
\newenvironment{conditions}{\begin{conditionss} \em}{\end{conditionss}}
\newenvironment{question}{\begin{questionn}}{\end{questionn}}
\newenvironment{conjecture}{\begin{conjecturee}}{\end{conjecturee}}

\title{A propositional system induced by Japaridze's approach to IF logic\thanks{Supported by National Natural Science Foundation of China (61303030) and the Fundamental Research Funds for the Central Universities of China (K5051370023).}}
\author{Wenyan Xu\\
{\it\small School of Mathematics and Statistics, Xidian University, Xi'an 710071, China}}
\date{}
\maketitle

\begin{abstract} Cirquent calculus is a new proof-theoretic and semantic approach introduced  for the needs of computability logic by G.Japaridze, who  also showed that, through cirquent calculus, one can capture, refine and generalize independence-friendly (IF) logic. Specifically, the approach allows us to account for independence from propositional connectives in the same spirit as the traditional IF logic accounts for independence from quantifiers. Japaridze's treatment of IF logic, however, was purely semantical, and no deductive system was proposed. The present paper constructs a  formal system sound and complete w.r.t. the propositional fragment of Japaridze's cirquent-based semantics for IF logic. Such a system can thus be considered an axiomatization of purely propositional IF logic in its full generality.

\end{abstract}

\noindent {\em MSC}: primary: 03B47; secondary: 03B70; 68Q10; 68T27; 68T15.

\

\noindent {\em Keywords}: Computability logic; Cirquent calculus; IF logic.

%\tableofcontents

\section{Introduction}\label{ssintr}
%\marginpar{ssintr}

{\em Cirquent calculus} is a new proof-theoretic and semantic approach introduced by G.Japaridze in \cite{Cir} and further developed in \cite{deep,fromto,taming1,taming2,wenyan1,wenyan2,wenyan3}. It was proposed for the needs of his computability logic (CoL) \cite{Jap03,beginning}. Unlike the more traditional proof theories that manipulate tree-like objects such as formulas, cirquent calculus deals with circuit-style objects termed {\em cirquents}. The main characteristic feature of the latter is allowing (one or another sort of) {\em sharing} of subcomponents between different components. Due to sharing, cirquent calculus has higher expressiveness and efficiency. For instance, as shown in \cite{deep}, the analytic cirquent calculus system CL8 achieves an exponential speedup of proofs over the classical analytic systems.

A qualitative generalization of the concept of cirquents was made in \cite{fromto}, where the idea of {\em clustering} of  propositional connectives was introduced. Intuitively, clusters  are   switch-style devices that combine tuples of individual disjunctive or conjunctive gates   in a parallel way ---  in a way where the choice  ({\em left} or {\em right}) of an argument is shared between all members. It was showed semantically in \cite{fromto} that, through cirquents with clustered  connectives  (and also quantifiers as generalized connectives), one can capture, refine and generalize the well known {\em independence-friendly (IF) logic}.\footnote{Clustering only disjunctions is sufficient for the basic IF logic, while the so called {\em extended IF logic} requires clustering both disjunctions and conjunctions.} The latter, introduced by J. Hintikka and G. Sandu \cite{IF} and further developed \cite{slash1,slash2,IFBOOK,tul,Van,A.P.,S.andP.} by a number of authors,  is a conservative extension of classical first-order logic,  allowing one to express independence relations between quantifiers. The past attempts (cf. \cite{A.P.,S.andP.}) to apply the same ideas to propositional connectives and thus develop IF logic at the propositional (as opposed to first-order) level, however, have remained limited only to certain special syntactic fragments of the language.

Japaridze's treatment of IF logic in \cite{fromto} was purely semantical, and no deductive system was proposed.
In this paper, we axiomatically construct a cirquent calculus system called {\bf $IF_p$}, with clustered disjunctive connectives, for propositional IF logic. Such a system is sound and complete w.r.t. the propositional fragment of Japaridze's cirquent-based semantics and can thus be considered an axiomatization of purely propositional, non-extended IF logic in its full generality.

\section{Preliminaries}
In this section we reproduce the basic relevant concepts from \cite{fromto}.
An interested reader may want to consult \cite{fromto} for additional
explanations, illustrations, examples and insights.

Our propositional language has infinitely many  {\bf atoms}, for which $p,q,r,s,\ldots$ will be used as metavariables. An atom $p$ and its negation $\neg p$ are called {\bf literals}.
A {\bf formula} means one of the language of classical propositional logic, built from
literals and the binary connectives $\wedge,\vee$ in the standard way.
Thus, all formulas are required to be in negation normal form. If we write
$A\rightarrow B$, it  is to be understood as an abbreviation of $\neg A\vee B$. And $\neg$, when applied to anything other than an atom, should be  understood as an abbreviation defined by $\neg\neg A=A$, $\neg(A\wedge B)=\neg A\vee\neg B$ and $\neg(A\vee B)=\neg A\wedge\neg B$.

A {\bf cirquent} is a formula together with a partition of the set of all occurrences of $\vee$ into subsets, called {\bf clusters}.\footnote{The concept of cirquents considered in cirquent calculus is more general than the one defined here. See \cite{fromto}.} With each cluster is associated a unique positive integer called its {\bf ID}. IDs serve as identifiers for clusters, and we will simply say  ``cluster $k$" to mean ``the cluster whose ID is $k$''.
%Formulas can (and will) be seen as special cases of cirquents, namely, the cases where all clusters are singletons.

One way to represent cirquents is to do so graphically,     using arcs to indicate the ``clusteral affiliations"  as  in the following figure:
\begin{center}
\begin{picture}(20,30)
\put(-90,0){$\bigl((p\vee\neg p)\wedge(p\vee\neg p)\bigr)\vee\bigl((q\vee r)\wedge(p\vee \neg q)\bigr)$}
\put(-47,17){\line(-5,-2){23}}\put(-47,17){\line(5,-2){23}}\put(-49,19){\footnotesize$1$}
\put(51,17){\line(-2,-1){19}}\put(51,17){\line(2,-1){19}}\put(49,19){\footnotesize$3$}
\put(5,19){\footnotesize$2$}\put(7,17){\line(0,-1){10}}
\end{picture}
\end{center}
For space efficiency considerations, in this paper we will instead be writing cirquents just like formulas, only  with every occurrence of $\vee$ indexed with the ID of the cluster to which the occurrence belongs. So, for instance, the above cirquent will be simply written as
$\bigl((p\vee_1\neg p)\wedge(p\vee_1\neg p)\bigr)\vee_2\bigl((q\vee_3 r)\wedge(p\vee_3 \neg q)\bigr)$.

 A cirquent $\mathcal{C}$ is said to be {\bf classical} iff all of its clusters are singletons. We shall identify such a cirquent with the formula of classical logic obtained from it by simply deleting all cluster IDs,
 %(which are not really relevant),
 i.e. replacing each $\vee_k$ (whatever $k$) with just $\vee$.

We will be using the term {\bf oconnective} to
refer to a connective together with a particular occurrence of it in a cirquent.
%Similarly, the terms ``oliteral", ``oatom", etc. will be used in this paper to refer to the corresponding entities together with particular occurrences.

An {\bf interpretation} (or {\bf model}) is a function $^{\ast}$ that sends each atom $p$ to one of the values $p^{\ast}\in\{\top,\bot\}$, and extends to all literals by stipulating that $(\neg p)^{\ast}=\top$ iff $p^{\ast}=\bot$.

A {\bf metaselection} is a function $f:\{1,2,3,\ldots\}\rightarrow\{\mbox{\em left,right}\}$.
%\footnote{What we here (for simplicity) call ``selection" is  called ``metaselection" in \cite{fromto}, with the word ``selection'' reserved for another version of the concept. There are some other minor technical differences not worth mentioning.}
Given a cirquent $\mathcal{C}$ and a metaselection $f$, the {\bf resolvent} of a disjunctive subcirquent $\mathcal{A}\vee_k \mathcal{B}$ of  $\mathcal{C}$ is defined to be $\mathcal{A}$ if $f(k)=\mbox{\em left}$, and $\mathcal{B}$ if $f(k)=\mbox{\em right}$.

Let $\mathcal{C}$ be a cirquent, $^{\ast}$ an interpretation, and $f$ a metaselection. In this context, with ``metatrue" to be read as ``{\bf metatrue w.r.t. $(^{\ast},f)$}", we say that:
\begin{itemize}
\item A literal $L$ of $\mathcal{C}$ is metatrue iff $L^{\ast}=\top$.
\item A subcirquent $\mathcal{A}\vee_k\mathcal{B}$ of $\mathcal{C}$  is metatrue iff so is its resolvent.
\item A subcirquent  $\mathcal{A}\wedge\mathcal{B}$ of $\mathcal{C}$ is metatrue iff so are both of its conjuncts.
\end{itemize}
Next, we say that $\mathcal{C}$ is {\bf true} under the interpretation $^*$ (in the model $^*$), or simply that $\mathcal{C}^*$ is true, iff there is a metaselection $f$ such that $\mathcal{C}$ is metatrue w.r.t. $(^{\ast},f)$.
Finally, we say that $\mathcal{C}$ is {\bf valid} iff it is true under every interpretation (in every model).

Note that, when $\mathcal{C}$ is a formula, i.e. a cirquent where all clusters are singletons, $\mathcal{C}$ is valid iff it is valid (tautological) in the sense of classical logic. And classical truth of a formula under an interpretation $^*$ means nothing but existence of a selection $f$ such that the formula is true in our sense w.r.t. $(^{\ast},f)$. So, classical logic is nothing but the conservative fragment of our logic obtained by only allowing formulas in the language.

Each cirquent $C$ with $n$ atoms can be seen to be an $n$-ary truth function (even though expressed in a very unusual way) and, for this reason,  $C$ can be represented by a disjunctive-normal-form formula $F$ of classical logic obtained in the standard way from the truth table characterizing $C$.
However, the formula $F$ obtained this way from $C$ will generally be exponentially bigger than $C$. We believe there is no translation from cirquents to equivalent Boolean formulas that only creates polynomial size differences, even though this needs to be proven, of course. Overall, an answer to the question whether our system yields a greater expressive power than classical propositional logic depends on who is asked. A philosopher would probably say ``No'', a computer scientist say ``Yes'', and a mathematician  say that it depends on how ``expressive power'' is precisely defined.

The rest of this section is not technically relevant to the main results of the present paper and is only meant  for those who are familiar with IF logic but not with \cite{fromto}.\footnote{As for those unfamiliar with IF logic, they may want to consult \cite{IFBOOK} or \cite{tul}.} In an attempt to understand what all of the above has to do with IF logic, consider the formula
\begin{equation}\label{wenzi0}
\forall x(\exists y/\forall x)\exists z\hspace{2pt} p(x,y,z)
\end{equation}
with its standard meaning. According to the latter, given any object $x$, two objects $y$ and $z$ can be chosen so that $p(x,y,z)$ is true, with the modifier `$\forall x$' attached to `$\exists y$' indicating that here $y$ can be chosen independently from (without any knowledge of)  $x$. Assuming that the universe of discourse is $\{1,2\}$, (\ref{wenzi0}) can just as well be (re-)written as

\[\begin{array}{c}
\hspace{-3pt}\Bigl(\bigl(p(1,1,1)\vee^z p(1,1,2)\bigr)\vee^y\hspace{-3pt}/\hspace{-3pt}\wedge^x \bigl(p(1,2,1)\vee^z p(1,2,2)\bigr)\Bigr)\\ \wedge^x \\
\Bigl(\bigl(p(2,1,1)\vee^z p(2,1,2)\bigr)\vee^y\hspace{-3pt}/\hspace{-3pt}\wedge^x \bigl(p(2,2,1)\vee^z p(2,2,2)\bigr)\Bigr),
\end{array} \]
which, after further rewriting $p(1,1,1),p(1,1,2),\ldots$ as the more compact $p,q,\ldots$, is the propositional formula
\begin{equation} \label{wenzi} \begin{array}{c}
\Bigl(\bigl(p\vee^z q\bigr)\vee^y\hspace{-3pt}/\hspace{-3pt}\wedge^x \bigl(r\vee^z s\bigr)\Bigr) \wedge^x
\Bigl(\bigl(t\vee^z u\bigr)\vee^y\hspace{-3pt}/\hspace{-3pt}\wedge^x \bigl(v\vee^z w\bigr)\Bigr).
\end{array} \end{equation}
Here we have turned $\forall x$ into $\wedge^x$, $\exists y$ into $\vee^y$ and $\exists z$ into $\vee^z$, with the superscript in each case used just to remind us from which quantifier each oconnective was obtained, and $\vee^y\hspace{-1pt}/\hspace{-1pt}\wedge^x$ indicating that the $y$-superscripted disjunction is independent of the $x$-superscripted conjunction. Now, Japaridze's recipes (see \cite{fromto}, Descriptions 7.4 and 7.5) translate (\ref{wenzi}) into the following cirquent:
%\marginpar{wenzi2}
\begin{equation} \label{wenzi2} \begin{array}{c}
\Bigl(\bigl(p\vee_2 q\bigr)\vee_1 \bigl(r\vee_3 s\bigr)\Bigr)\wedge\Bigl(\bigl(t\vee_4 u\bigr)\vee_1 \bigl(v\vee_5 w\bigr)\Bigr).
\end{array} \end{equation}
Note that cluster $1$ contains two disjunctive oconnectives --- namely, those originating from $\exists y$, and all other clusters are singletons. It is left as an exercise for the reader to convince himself or herself that, in any given model (interpretation) $^*$, (\ref{wenzi2}) is true in our sense if and only if (\ref{wenzi}) is true in the sense of IF logic. Well, the present case is a ``lucky'' case because we easily understand what ``true in the sense of IF logic'' means for (\ref{wenzi}) --- after all, (\ref{wenzi}) originates from (and will be handled in the same way as) the first-order (\ref{wenzi0}). As an example of an ``unlucky'' case, consider the cirquent
\begin{equation} \label{wenzi3} \begin{array}{c}
(p\vee_1 q)\wedge\bigl((r\vee_1 s)\wedge q\bigr).
\end{array} \end{equation}
It is just as meaningful from the point of view of our semantics as any other cirquent, including (\ref{wenzi2}). An attempt to express the same in the traditional formalism of IF logic apparently yields something like
\begin{equation} \label{wenzi33} \begin{array}{c}
(p\vee\hspace{-3pt}/\hspace{-3pt}\wedge^x q)\wedge^x\bigl((r\vee\hspace{-3pt}/\hspace{-3pt}\wedge^x s)\wedge^y q\bigr).
\end{array} \end{equation}
Unlike (\ref{wenzi}), however, (\ref{wenzi33}) is problematic for the traditional semantical approaches (the ones based on imperfect information games) to IF logic. Namely, because of a problem called signaling, it is far from clear how its truth should be understood.

If the connections and differences between our present semantics and that of IF logic are still not clear, see the first 6 sections of \cite{fromto} for more explanations, discussions and examples.

%At the end of this section, let us see why cirquent calculus accounts for independence from disjunctions and why system $IF_p$ is the fragment of IF logic.

\section{Main results}
\subsection{System $IF_p$ introduced}
Throughout the rest of this paper, for convenience of descriptions, we will be always omitting the IDs for singleton clusters when writing cirquents.
 %and by ``a cirquent" we mean the one with the indices (IDs) for singleton clusters omitted.
 Note that, this way, all classical cirquents mechanically turn into (rather than are just identified with) formulas of classical logic.
  %and thus the axioms of system $IF_p$ are identified with classical tautologies. For instance, $(p\vee q)\vee \neg p$ is an axiom of $IF_p$, because it is a classical tautology.

As will be seen shortly, the inference rules of our system $IF_p$
modify cirquents at any level rather than only around the root. Thus, $IF_p$ is in fact a {\em deep inference} system, in the style of \cite{deepinference}. This explains our borrowing some notation from the Calculus of Structures.
Namely, we
 will be using $\Phi\{\}$ or $\Psi\{\}$ to denote any cirquent where a vacancy (``hole'') $\{\}$ appears in the place of a subcirquent. The vacancy $\{\}$ can be filled with any cirquent. For example, if $\Phi\{\} = (p\vee_1 q)\vee_1(\{\}\wedge q)$, then $\Phi\{\neg p\} = (p\vee_1 q)\vee_1(\neg p\wedge q)$, $\Phi\{q\} = (p\vee_1 q)\vee_1(q\wedge q)$, and $\Phi\{p\vee_1 q\} = (p\vee_1 q)\vee_1((p\vee_1 q)\wedge q)$. Multiple vacancies are also allowed and they are treated similarly.\vspace{2mm}

The {\bf axioms} of  $IF_p$ are all classical cirquents that (seen as formulas)
 are tautologies of classical logic.
We schematically represent the {\bf rules of inference} of $IF_p$ in  Figure 1, where $\mathcal{A,B,C,D}$ stand for any cirquents and $\circ$ is a variable that stands for either $\wedge$ or $\vee$ or $\vee_l$ for some (whatever) particular index $l$ such that cluster $l$ is a non-singleton one. It is important to point out that, in each rule, {\em all} occurrences of $\circ$ stand for the same object ($\wedge$, $\vee$ or
$\vee_l$).
%Note that, when $\circ$ is $\vee$, each occurrence of $\circ$ in the premise (resp. the conclusion) of a rule stands for a particular $\vee$ such that $\vee$ is in a singleton cluster in the premise (resp. the conclusion).
%Certain necessary explanations of their meanings and examples of applications follow.
\vspace{2mm}

\begin{center}
\begin{picture}(80,170)(0,35)
\put(-45,0){\begin{picture}(80,200)\put(-79,190){$\Phi\bigl\{\Psi\{\mathcal{A}\}\vee_k\mathcal{C}\bigr\}$}
\put(-89,184){\line(1,0){85}}\put(0,183){\footnotesize\bf Rule I (left)}
\put(-90,172){$\Phi\bigl\{\Psi\{\mathcal{A}\vee_k\mathcal{B}\}\vee_k\mathcal{C}\bigr\}$}\end{picture}}

\put(150,0){\begin{picture}(80,200)\put(-79,190){$\Phi\bigl\{\mathcal{C}\vee_k\Psi\{\mathcal{A}\}\bigr\}$}
\put(-89,184){\line(1,0){85}}\put(0,183){\footnotesize\bf Rule I (right)}
\put(-90,172){$\Phi\bigl\{\mathcal{C}\vee_k\Psi\{\mathcal{B}\vee_k\mathcal{A}\}\bigr\}$}\end{picture}}

\put(-43,-50){\begin{picture}(80,200)\put(-95,190){$\Phi\bigl\{(\mathcal{A}\circ\mathcal{C})\vee_k(\mathcal{B}\circ\mathcal{C})\bigr\}$}
\put(-95,184){\line(1,0){95}}\put(3,183){\footnotesize\bf Rule II (left)}
\put(-83,172){$\Phi\bigl\{(\mathcal{A}\vee_k\mathcal{B})\circ\mathcal{C}\bigr\}$}\end{picture}}

\put(152,-50){\begin{picture}(80,200)\put(-93,190){$\Phi\bigl\{(\mathcal{C}\circ\mathcal{A})\vee_k(\mathcal{C}\circ\mathcal{B})\bigr\}$}
\put(-93,184){\line(1,0){95}}\put(5,183){\footnotesize\bf Rule II (right)}
\put(-82,172){$\Phi\bigl\{\mathcal{C}\circ(\mathcal{A}\vee_k\mathcal{B})\bigr\}$}\end{picture}}

\put(60,-100){\begin{picture}(80,200)\put(-95,190){$\Phi\bigl\{(\mathcal{A}\circ\mathcal{C})\vee_k(\mathcal{B}\circ\mathcal{D})\bigr\}$}
\put(-97,184){\line(1,0){104}}\put(10,183){\footnotesize\bf Rule III}
\put(-97,172){$\Phi\bigl\{(\mathcal{A}\vee_k\mathcal{B})\circ(\mathcal{C}\vee_k\mathcal{D})\bigr\}$}\end{picture}}

%\put(150,-100){\begin{picture}(80,200)\put(-94,190){$\Phi\bigl\{(\mathcal{A}\vee_i\mathcal{C})\vee_k(\mathcal{B}\vee_j\mathcal{D})\bigr\}$}
%\put(-94,184){\line(1,0){107}}\put(17,183){\footnotesize\bf Rule IV ($\vee$)}
%\put(-94,172){$\Phi\bigl\{(\mathcal{A}\vee_k\mathcal{B})\vee_l(\mathcal{C}\vee_k\mathcal{D})\bigr\}$}\end{picture}}

\put(-32,43){{\bf Figure 1:} The rules of $IF_p$}
\end{picture}
\end{center}

In each application of these rules, we call the oconnective(s) $\vee_k$ in the premise (resp. conclusion), as shown in Figure 1, the {\bf key oconnecitve(s)} of this application in the premise (resp. conclusion).

A {\bf proof} of a cirquent $\mathcal{A}$ in $IF_p$, as expected, is a sequence of cirquents such that the first cirquent in the sequence is an axiom of $IF_p$, the last cirquent is  $\mathcal{A}$, and every cirquent, except the axiom, follows from the preceding cirquent by one of the rules of $IF_p$. When such a proof exists, $\mathcal{A}$ is said to be {\bf provable} in $IF_p$.
Below is an example of a proof:
\begin{center}
\begin{picture}(80,110)(25,3)
\put(-60,100){\line(1,0){205}}\put(148,98){\bf\footnotesize Axiom}
\put(-53,90){$((q\wedge p)\vee(p\wedge\neg q))\vee((q\wedge\neg p)\vee(\neg p\wedge\neg q))$}
\put(-60,85){\line(1,0){205}}\put(148,83){\bf\footnotesize Rule III}
\put(-57,75){$((q\wedge p)\vee_2(q\wedge\neg p))\vee((p\wedge\neg q)\vee_2(\neg p\wedge\neg q))$}
\put(-60,70){\line(1,0){205}}\put(148,68){\bf\footnotesize Rule II (left)}
\put(-42,60){$((q\wedge p)\vee_2(q\wedge\neg p))\vee((p\vee_2\neg p)\wedge\neg q)$}
\put(-60,55){\line(1,0){205}}\put(148,53){\bf\footnotesize Rule II (right)}
\put(-30,45){$(q\wedge(p\vee_2\neg p))\vee((p\vee_2\neg p)\wedge\neg q)$}
\put(-60,40){\line(1,0){205}}\put(148,38){\bf\footnotesize Rule I (right)}
\put(-45,30){$(q\wedge(p\vee_2\neg p))\vee_1((p\vee_2\neg p)\wedge(s\vee_1\neg q))$}
\put(-60,25){\line(1,0){205}}\put(148,23){\bf\footnotesize Rule I (left)}
\put(-60,15){$((q\vee_1 r)\wedge(p\vee_2\neg p))\vee_1((p\vee_2\neg p)\wedge(s\vee_1\neg q))$}
\end{picture}
\end{center}

\begin{lemma}\label{yan1}
Given an  interpretation $^{\ast}$ and  a metaselection $f$, a cirquent $\Phi\{\mathcal{A}\vee_k\mathcal{B}\}$ is metatrue w.r.t. $(^{\ast},f)$ iff $f(k)=\mbox{\em left}$ (resp. $f(k)=\mbox{\em right}$)  and
the cirquent $\Phi\{\mathcal{A}\}$ (resp. $\Phi\{\mathcal{B}\}$) is metatrue w.r.t. $(^{\ast},f)$.
\end{lemma}

\begin{proof}
We prove the proposition by induction on the number of oconnectives of $\Phi\{\}$.

For the basis, assume that the number of oconnectives of $\Phi\{\}$ is $0$. Then $\Phi\{\mathcal{A}\vee_k\mathcal{B}\}=\mathcal{A}\vee_k\mathcal{B}$, $\Phi\{\mathcal{A}\}=\mathcal{A}$ and $\Phi\{\mathcal{B}\}=\mathcal{B}$. By the definition of metatruth, we immediately have $\mathcal{A}\vee_k\mathcal{B}$ is metatrue w.r.t. $(^{\ast},f)$ if and only if $f(k)=\mbox{\em left}$ (resp. $f(k)=\mbox{\em right}$) and its resolvent $\mathcal{A}$ (resp. $\mathcal{B}$) is so.

Now (induction hypothesis) assume that the proposition holds when the number of oconnectives of $\Phi\{\}$ is $n$.
We want to show that the proposition still holds when the number of oconnectives of $\Phi\{\}$ is $n+1$. The following three cases (i), (ii), (iii) need to be considered:

(i) Assume that the main connective of $\Phi\{\}$ is $\wedge$, namely, $\Phi\{\mathcal{A}\vee_k\mathcal{B}\}=\Psi\{\mathcal{A}\vee_k\mathcal{B}\}\wedge\mathcal{C}$  for some cirquent $\mathcal{C}$ (the other possibility
$\Phi\{\mathcal{A}\vee_k\mathcal{B}\}= \mathcal{C} \wedge\Psi\{\mathcal{A}\vee_k\mathcal{B}\}$ is similar).
By the definition of metatruth, $\Phi\{\mathcal{A}\vee_k\mathcal{B}\}$ is metatrue w.r.t. $(^{\ast},f)$ iff both $\Psi\{\mathcal{A}\vee_k\mathcal{B}\}$ and $\mathcal{C}$ are so. But, by the induction hypothesis, $\Psi\{\mathcal{A}\vee_k\mathcal{B}\}$ is metatrue w.r.t. $(^{\ast},f)$ iff $f(k)=\mbox{\em left}$ (resp. $f(k)=\mbox{\em right}$) and $\Psi\{\mathcal{A}\}$ (resp. $\Psi\{\mathcal{B}\}$) is metatrue w.r.t. $(^{\ast},f)$. Therefore, $\Phi\{\mathcal{A}\vee_k\mathcal{B}\}$ is metatrue w.r.t. $(^{\ast},f)$ iff $f(k)=\mbox{\em left}$ (resp. $f(k)=\mbox{\em right}$) and $\Psi\{\mathcal{A}\}\wedge\mathcal{C}=\Phi\{\mathcal{A}\}$ (resp. $\Psi\{\mathcal{B}\}\wedge\mathcal{C}=\Phi\{\mathcal{B}\}$) is metatrue w.r.t. $(^{\ast},f)$.

(ii) Assume that the main connective of $\Phi\{\}$ is $\vee$, namely, $\Phi\{\mathcal{A}\vee_k\mathcal{B}\}=\Psi\{\mathcal{A}\vee_k\mathcal{B}\}\vee\mathcal{C}$ (the case $\Phi\{\mathcal{A}\vee_k\mathcal{B}\}=\mathcal{C}\vee\Psi\{\mathcal{A}\vee_k\mathcal{B}\}$ is similar). This means that the displayed (main) occurrence of
 $\vee$ is in a singleton cluster $i$ for some positive integer $i\neq k$.  Two (sub)cases are to be further considered here.

{\it Case (a)}: $f(i)=\mbox{\em right}$.  Then $\Psi\{\mathcal{A}\vee_k\mathcal{B}\}\vee\mathcal{C}$ is metatrue w.r.t. $(^{\ast},f)$ iff its resolvent $\mathcal{C}$ is so. But exactly the same (and for the same reason) holds for both $\Psi\{\mathcal{A}\}\vee\mathcal{C}$ and $\Psi\{\mathcal{B}\}\vee\mathcal{C}$.  Consequently, vacuously adding ``$f(k)=\ldots$'', we arrive at the desired conclusion that $\Psi\{\mathcal{A}\vee_k\mathcal{B}\}\vee\mathcal{C}$ is metatrue w.r.t. $(^{\ast},f)$ iff $f(k)=\mbox{\em left}$ (resp. $f(k)=\mbox{\em right}$) and $\Psi\{\mathcal{A}\}\vee\mathcal{C}$ (resp. $\Psi\{\mathcal{B}\}\vee\mathcal{C}$) is
 metatrue w.r.t. $(^{\ast},f)$.

{\it Case (b)}: $f(i)=\mbox{\em left}$. Then $\Phi\{\mathcal{A}\vee_k\mathcal{B}\}$ is metatrue w.r.t. $(^{\ast},f)$ iff its resolvent $\Psi\{\mathcal{A}\vee_k\mathcal{B}\}$ is so, which, in turn (by the induction hypothesis), is the case  iff  $f(k)=\mbox{\em left}$ (resp. $f(k)=\mbox{\em right}$) and $\Psi\{\mathcal{A}\}$ (resp. $\Psi\{\mathcal{B}\}$) is metatrue w.r.t. $(^{\ast},f)$. This, in turn, is the case iff $f(k)=\mbox{\em left}$ (resp. $f(k)=\mbox{\em right}$) and $\Phi\{\mathcal{A}\}$ (resp. $\Phi\{\mathcal{B}\}$) is metatrue w.r.t. $(^{\ast},f)$. Hence the desired conclusion holds.

(iii) Assume that the main connective of $\Phi\{\}$ is $\vee_l$, for some (whatever) particular index $l$ such that cluster $l$ is not a singleton. Namely, assume  $\Phi\{\mathcal{A}\vee_k\mathcal{B}\}=\Psi\{\mathcal{A}\vee_k\mathcal{B}\}\vee_l\mathcal{C}$ (the case
$\Phi\{\mathcal{A}\vee_k\mathcal{B}\}=\mathcal{C}\vee_l\Psi\{\mathcal{A}\vee_k\mathcal{B}\}$ is similar).
If $l\neq k$, then we can employ an essentially the same argument as the one used in (ii).  And if
 $l=k$, then the case is even simpler, so we leave details to the reader.
\end{proof}

\begin{lemma}\label{niu}
All rules of $IF_p$ preserve truth in both top-down and bottom-up directions.
\end{lemma}

\begin{proof} Pick an arbitrary interpretation $^*$.
%For each rule of the  system, we want to show

{\em Rule I}: Here we will only look at Rule I (left),  with Rule I (right) being similar. Consider an arbitrary  metaselection $f$. By Lemma \ref{yan1},
the premise is metatrue w.r.t. $(^{\ast},f)$ iff $f(k)=\mbox{\em left}$ (resp. $f(k)=\mbox{\em right}$) and  $\Phi\bigl\{\Psi\{\mathcal{A}\}\bigr\}$ (resp. $\Phi\{\mathcal{C}\}$)  is metatrue w.r.t. $(^{\ast},f)$. But exactly the same is the case for the conclusion as well (only, now Lemma \ref{yan1} needs to be applied twice). Thus, the premise is metatrue w.r.t. $(^{\ast},f)$ iff so is the conclusion. This, of course, implies that the premise is true under $^*$ iff so is the conclusion, as desired.

{\em Rule II}:  Again, we will only consider Rule II (left),  with Rule II (right) being similar.
 We want to show that the premise $\Phi\bigl\{(\mathcal{A}\circ\mathcal{C})\vee_k(\mathcal{B}\circ\mathcal{C})\bigr\}$ is true under $^{\ast}$ iff so is the conclusion $\Phi\bigl\{(\mathcal{A}\vee_k\mathcal{B})\circ\mathcal{C}\bigr\}$.
%Given an arbitrary interperation $^*$, the premise $\Phi\bigl\{(\mathcal{A}\circ\mathcal{C})\vee_k(\mathcal{B}\circ\mathcal{C})\bigr\}$ is true under $^{\ast}$ iff there is a metaselection $f$ such that $\Phi\bigl\{(\mathcal{A}\circ\mathcal{C})\vee_k(\mathcal{B}\circ\mathcal{C})\bigr\}$ is metatrue w.r.t. $(^{\ast},f)$. By Lemma \ref{yan1}, $\Phi\bigl\{(\mathcal{A}\circ\mathcal{C})\vee_k(\mathcal{B}\circ\mathcal{C})\bigr\}$ is metatrue w.r.t. $(^{\ast},f)$ iff $f(k)=\mbox{\em left}$ (resp. $f(k)=\mbox{\em right}$) and $\Phi\{\mathcal{A}\circ\mathcal{C}\}$ (resp. $\Phi\{\mathcal{B}\circ\mathcal{C}\}$) is so.
The following two cases need to be considered here.

{\em Case (a)}: $\circ$ is $\wedge$ or $\vee_l$. In this case, by Lemma \ref{yan1} applied to both the premise
%$\Phi\bigl\{(\mathcal{A}\circ\mathcal{C})\vee_k(\mathcal{B}\circ\mathcal{C})\bigr\}$
and the conclusion,
% $\Phi\bigl\{(\mathcal{A}\vee_k\mathcal{B})\circ\mathcal{C}\bigr\}$,
 we immediately get that, for an arbitrary metaselection $f$,  the premise is metatrue w.r.t. $(^{\ast},f)$ iff so is the conclusion. Hence, the premise  is true under $^*$ iff so is the conclusion.

{\em Case (b)}: $\circ$ is $\vee$. Assume that, after restoring the (otherwise always omitted) singleton-cluster IDs, the premise
%$\Phi\bigl\{(\mathcal{A}\circ\mathcal{C})\vee_k(\mathcal{B}\circ\mathcal{C})\bigr\}$
is $\Phi\bigl\{(\mathcal{A}\vee_i\mathcal{C})\vee_k(\mathcal{B}\vee_j\mathcal{C})\bigr\}$
and the conclusion
%$\Phi\bigl\{(\mathcal{A}\vee_k\mathcal{B})\circ\mathcal{C}\bigr\}$
is $\Phi\bigl\{(\mathcal{A}\vee_k\mathcal{B})\vee_m\mathcal{C}\bigr\}$, where clusters $i,j,m$ are singletons.
Let $\{l_1,\ldots,l_n\}$ be the collection of all singleton clusters in the subcirquent $\mathcal{C}$ of the conclusion.
And let $\{l'_1,\ldots,l'_n\}$ (resp. $\{l''_1,\ldots,l''_n\}$) be the collection of all singleton clusters in the left (resp. right) occurrence of $\mathcal{C}$ in the premise satisfying the condition that, for any $h\in\{1,\ldots,n\}$, $l'_h$ (resp. $l''_h$) occurs in this $\mathcal{C}$ at the same place as $l_h$ occurs in the $\mathcal{C}$ part of the conclusion.

Suppose that the conclusion $\Phi\bigl\{(\mathcal{A}\vee_k\mathcal{B})\vee_m\mathcal{C}\bigr\}$ is true under $^*$. Then there is a metaselection $f$ such that $\Phi\bigl\{(\mathcal{A}\vee_k\mathcal{B})\vee_m\mathcal{C}\bigr\}$ is metatrue w.r.t. $(^{\ast},f)$. This, by Lemma \ref{yan1}, implies that  $f(k)=\mbox{\em left}$ (resp. $f(k)=\mbox{\em right}$) and $\Phi\{\mathcal{A}\vee_m\mathcal{C}\}$ (resp. $\Phi\{\mathcal{B}\vee_m\mathcal{C}\}$) is  metatrue w.r.t. $(^{\ast},f)$.
Let $g$ be a metaselection satisfying the conditions that $g(i)=g(j)=f(m)$, $g(l'_h)=g(l''_h)=f(l_h)$ for any $h\in\{1,\ldots,n\}$ and $g$ agrees with $f$ on all other clusters. Then we have $g(k)=\mbox{\em left}$ (resp. $g(k)=\mbox{\em right}$) and $\Phi\{\mathcal{A}\vee_i\mathcal{C}\}$ (resp. $\Phi\{\mathcal{B}\vee_j\mathcal{C}\}$) is  metatrue w.r.t. $(^{\ast},g)$. This, again by Lemma \ref{yan1}, implies that the premise $\Phi\bigl\{(\mathcal{A}\vee_i\mathcal{C})\vee_k(\mathcal{B}\vee_j\mathcal{C})\bigr\}$ is metatrue w.r.t. $(^{\ast},g)$. Therefore, the premise is true under $^*$.

Now suppose that the premise $\Phi\bigl\{(\mathcal{A}\vee_i\mathcal{C})\vee_k(\mathcal{B}\vee_j\mathcal{C})\bigr\}$ is true under  $^*$, meaning that there is a metaselection $u$ such that $\Phi\bigl\{(\mathcal{A}\vee_i\mathcal{C})\vee_k(\mathcal{B}\vee_j\mathcal{C})\bigr\}$ is metatrue w.r.t. $(^{\ast},u)$. By Lemma \ref{yan1}, this implies that $u(k)=\mbox{\em left}$ (resp. $u(k)=\mbox{\em right}$) and $\Phi\{\mathcal{A}\vee_i\mathcal{C}\}$ (resp. $\Phi\{\mathcal{B}\vee_j\mathcal{C}\}$) is  metatrue w.r.t. $(^{\ast},u)$.
Let $v$ be a metaselection satisfying the following two conditions: if $u(k)=\mbox{\em left}$ (resp. $u(k)=\mbox{\em right}$), then $v(m)=u(i)$ (resp. $v(m)=u(j)$) and $v(l_h)=u(l'_h)$ (resp. $v(l_h)=u(l''_h)$) for any $h\in\{1,\ldots,n\}$; $v$ agrees with $u$ on all other clusters. Then, we have $v(k)=\mbox{\em left}$ (resp. $v(k)=\mbox{\em right}$) and $\Phi\{\mathcal{A}\vee_m\mathcal{C}\}$ (resp. $\Phi\{\mathcal{B}\vee_m\mathcal{C}\}$) is  metatrue w.r.t. $(^{\ast},v)$. This, by Lemma \ref{yan1}, implies that the conclusion $\Phi\bigl\{(\mathcal{A}\vee_k\mathcal{B})\vee_m\mathcal{C}\bigr\}$ is metatrue w.r.t. $(^{\ast},v)$. Hence, the conclusion is true under $^*$.

{\em Rule III}: We want to show that the premise $\Phi\bigl\{(\mathcal{A}\circ\mathcal{C})\vee_k(\mathcal{B}\circ\mathcal{D})\bigr\}$ is true under $^{\ast}$ iff so is the conclusion $\Phi\bigl\{(\mathcal{A}\vee_k\mathcal{B})\circ(\mathcal{C}\vee_k\mathcal{D})\bigr\}$. Two cases are to be considered.

{\em Case (a)}: $\circ$ is $\wedge$ or $\vee_l$. Consider an arbitrary metaselection $f$. It is not hard to see that, by Lemma \ref{yan1} applied twice, the conclusion $\Phi\bigl\{(\mathcal{A}\vee_k\mathcal{B})\circ(\mathcal{C}\vee_k\mathcal{D})\bigr\}$ is metatrue w.r.t. $(^{\ast},f)$ iff $f(k)=\mbox{\em left}$ (resp. $f(k)=\mbox{\em right}$) and $\Phi\{\mathcal{A}\circ\mathcal{C}\}$ (resp. $\Phi\{\mathcal{B}\circ\mathcal{D}\}$) is  metatrue w.r.t. $(^{\ast},f)$. This, in turn, is the case iff the premise $\Phi\bigl\{(\mathcal{A}\circ\mathcal{C})\vee_k(\mathcal{B}\circ\mathcal{D})\bigr\}$ is metatrue w.r.t. $(^{\ast},f)$. Therefore, the conclusion is true under $^{\ast}$ iff so is the premise.

{\em Case (b)}: $\circ$ is $\vee$. Assume that (after restoring the singleton-cluster IDs) the premise is $\Phi\bigl\{(\mathcal{A}\vee_i\mathcal{C})\vee_k(\mathcal{B}\vee_j\mathcal{D})\bigr\}$ and the conclusion is $\Phi\bigl\{(\mathcal{A}\vee_k\mathcal{B})\vee_m(\mathcal{C}\vee_k\mathcal{D})\bigr\}$, where clusters $i,j,m$ are singletons. As can be seen from the following two paragraphs, the present case is very similar to Case (b) of Rule II.

Suppose that the conclusion $\Phi\bigl\{(\mathcal{A}\vee_k\mathcal{B})\vee_m(\mathcal{C}\vee_k\mathcal{D})\bigr\}$ is true under $^*$. Then there is a metaselection $f$ such that $\Phi\bigl\{(\mathcal{A}\vee_k\mathcal{B})\vee_m(\mathcal{C}\vee_k\mathcal{D})\bigr\}$ is metatrue w.r.t. $(^{\ast},f)$. This, by Lemma \ref{yan1} applied twice, implies that  $f(k)=\mbox{\em left}$ (resp. $f(k)=\mbox{\em right}$) and $\Phi\{\mathcal{A}\vee_m\mathcal{C}\}$ (resp. $\Phi\{\mathcal{B}\vee_m\mathcal{D}\}$) is  metatrue w.r.t. $(^{\ast},f)$. Let $g$ be a metaselection satisfying the conditions that $g(i)=g(j)=f(m)$ and $g$ agrees with $f$ on all other clusters. Then we have $g(k)=\mbox{\em left}$ (resp. $g(k)=\mbox{\em right}$) and $\Phi\{\mathcal{A}\vee_i\mathcal{C}\}$ (resp. $\Phi\{\mathcal{B}\vee_j\mathcal{D}\}$) is  metatrue w.r.t. $(^{\ast},g)$. This, again by Lemma \ref{yan1}, implies that the premise $\Phi\bigl\{(\mathcal{A}\vee_i\mathcal{C})\vee_k(\mathcal{B}\vee_j\mathcal{D})\bigr\}$ is metatrue w.r.t. $(^{\ast},g)$. Therefore, the premise is true under $^*$.

Now suppose that the premise $\Phi\bigl\{(\mathcal{A}\vee_i\mathcal{C})\vee_k(\mathcal{B}\vee_j\mathcal{D})\bigr\}$ is true under  $^*$. Then there is a metaselection $u$ such that $\Phi\bigl\{(\mathcal{A}\vee_i\mathcal{C})\vee_k(\mathcal{B}\vee_j\mathcal{D})\bigr\}$ is metatrue w.r.t. $(^{\ast},u)$. By Lemma \ref{yan1}, this implies that $u(k)=\mbox{\em left}$ (resp. $u(k)=\mbox{\em right}$) and $\Phi\{\mathcal{A}\vee_i\mathcal{C}\}$ (resp. $\Phi\{\mathcal{B}\vee_j\mathcal{D}\}$) is  metatrue w.r.t. $(^{\ast},u)$. Let $v$ be a metaselection satisfying the following two conditions: when $u(k)=\mbox{\em left}$ (resp. $u(k)=\mbox{\em right}$), $v(m)=u(i)$ (resp. $v(m)=u(j)$); $v$ agrees with $u$ on all other clusters. Then, we have $v(k)=\mbox{\em left}$ (resp. $v(k)=\mbox{\em right}$) and $\Phi\{\mathcal{A}\vee_m\mathcal{C}\}$ (resp. $\Phi\{\mathcal{B}\vee_m\mathcal{D}\}$) is  metatrue w.r.t. $(^{\ast},v)$. This, by Lemma \ref{yan1} applied twice, implies that the conclusion $\Phi\bigl\{(\mathcal{A}\vee_k\mathcal{B})\vee_m(\mathcal{C}\vee_k\mathcal{D})\bigr\}$ is metatrue w.r.t. $(^{\ast},v)$. Hence, the conclusion is true under $^*$.
\end{proof}

\subsection{The soundness and completeness of $IF_p$}
\begin{theorem}
 A cirquent is valid if and only if it is provable in $IF_p$.
\end{theorem}
\begin{proof}
The soundness part is immediate, because the axioms are obviously valid and, by Lemma \ref{niu}, all rules preserve truth and hence validity. For the completeness part, consider an arbitrary cirquent $\mathcal{C}$ and assume it is valid. We want to show that $\mathcal{C}$ is provable in $IF_p$.\vspace{2mm}

In the context of a given cirquent $\mathcal{D}$, we define the {\bf level} of an oconnective $a$, denoted by $\mathcal{L}^{\mathcal{D}}(a)$, to be the total number of oconnectives $b$ such that $a$ is in the scope of $b$. An oconnective $b$ is a {\bf child} of an oconnective $a$ and $a$ is the {\bf parent} of $b$ when $b$ is in the scope of $a$ and $\mathcal{L}^{\mathcal{D}}(b) = \mathcal{L}^{\mathcal{D}}(a)+1$. The relations ``descendant" and ``ancestor" are the transitive closures of the relations ``child" and ``parent", respectively. The {\bf distance} between an oconnective $a$ and one of its descendants $b$ is defined to be the positive integer $k$ such that $k = \mathcal{L}^{\mathcal{D}}(b)-\mathcal{L}^{\mathcal{D}}(a)$; when the distance between $a$ and $b$ is less than the distance between $a$ and another descendant $c$ of $a$, we say that $b$ is {\bf nearer} to $a$ (or vice versa) than $c$ is. Next, for any two oconnectives $a$ and $b$, we denote their nearest common ancestor oconnective by  $\underline{ab}$.\vspace{2mm}

If our cirquent $\mathcal{C}$ is classical (i.e. every cluster of it is a singleton),  then  the validity of $\mathcal{C}$ can be seen to mean nothing but its validity in the sense of classical logic. So, in this case,  $\mathcal{C}$ is an axiom of $IF_p$ and hence is provable.

Now, for the rest of this proof, assume that $\mathcal{C}$ is not classical. We construct, bottom-up, a proof of $\mathcal{C}$ as follows.

First, repeat applying Rule I to the current (topmost in the so far constructed proof) cirquent until no longer possible.

Every time Rule I is applied, the current cirquent loses one pair of disjunctive oconnectives $a,b$ such that $a,b$ are in the same cluster and $b$ is a descendant of $a$ (or vice versa). So, sooner or later, we get a cirquent $\mathcal{C}_{1}$ where no descendant-ancestor pair of oconnectives shares the same cluster. Since Rule I preserves truth in the bottom-up direction (Lemma \ref{niu}), $\mathcal{C}_{1}$ is valid.

Let $k_1,k_2,\ldots,k_n$ ($n\geq 0$) be a list of all non-singleton clusters of  $\mathcal{C}_{1}$. Our construction of a proof of $\mathcal{C}$ continues upward from $\mathcal{C}_{1}$ as follows. \vspace{2mm}

Pick any cluster $k_i$ from the above list. Repeat the following three steps until cluster $k_i$ becomes a singleton cluster in the current (topmost) cirquent. Below  we use $\mathcal{D}$ to denote the current cirquent:\vspace{2mm}

{\em Step 1.} Pick any pair  $a,b$ of disjunctive oconnectives such that the following two conditions are satisfied: $a$, $b$ are both in cluster $k_i$; $\mathcal{L}^{\mathcal{D}}(\underline{ab})\geq\mathcal{L}^{\mathcal{D}}(\underline{cd})$ for any pair of oconnectives $c, d$ in cluster $k_i$. Set $m=2$ and $l=\mathcal{L}^{\mathcal{D}}(\underline{ab})$.\footnote{Here $m$ is a variable that records the number of elements of the collection of key oconnectives in the current cirquent. It is introduced into the process mainly for later difinitions and proofs.} \vspace{2mm}

{\em Step 2.} (a) Repeatedly perform the following two actions until $\mathcal{L}^{\mathcal{D}}(a)=l+1$: (i) Apply (bottom-up) Rule II to $\mathcal{D}$, with $a$ being the key oconnective of this application (in the conclusion); (ii) Rename the key oconnective of this application (in the premise) into $a$.

(b)  Repeatedly perform the following two actions until $\mathcal{L}^{\mathcal{D}}(b)=l+1$: (i) Apply (bottom-up) Rule II to $\mathcal{D}$, with $b$ being the key oconnective of this application (in the conclusion); (ii) Rename the key oconnective of this application (in the premise) into $b$.\vspace{2mm}

{\em Step 3.} Apply (bottom-up) Rule III to $\mathcal{D}$, with $a,b$ being the key oconnectives of this application (in the conclusion). Set $m=1$.\vspace{2mm}

During the above three-step procedure, the situation that some descendant together with its ancestor is in the same cluster will not emerge. To see why, note that there are no descendant-ancestor pairs within any given cluster in the current cirquent $\mathcal{D}$ at the beginning of the procedure. If so, step 1 will not give rise to the situation. Next,   applying Rule II to $\mathcal{D}$ in step 2 means that $\mathcal{L}^{\mathcal{D}}(a)$ (resp. $\mathcal{L}^{\mathcal{D}}(b)$) is greater than $\mathcal{L}^{\mathcal{D}}(\underline{ab})+1$; fourthly, during step 1 and step 2, $l$'s being maximal ensures that, when applying Rule II as Figure 1 shows to the current cirquent $\mathcal{D}$, the subcirquent $\mathcal{C}$ does not contain any oconnective $e$ that is also in cluster $k_i$ (otherwise, $\mathcal{L}^{\mathcal{D}}(\underline{ae})>\mathcal{L}^{\mathcal{D}}(\underline{ab})$ (resp. $\mathcal{L}^{\mathcal{D}}(\underline{be})>\mathcal{L}^{\mathcal{D}}(\underline{ab})$), which contradicts the conditions satisfied by $a,b$). Finally, based on the above conditions, applying Rule III in step 3 will not create any descendant-ancestor pairs within any given cluster, either.

Below we verify  that the above three-step procedure terminates in a finite number of steps.

For the current cirquent $\mathcal{D}$ at any given stage of the procedure --- hencefore we shall use $\mathcal D$ as (also) a name of that stage ---  we define the {\bf state} of $\mathcal {D}$ to be the four-tuple $(x,y,z,t)$ as follows, where $m_{\cal D}$, $a_{\cal D}$, $b_{\cal D}$ are the values of the corresponding three variables of the procedure at the beginning of stage $\cal D$:
\begin{itemize}
  \item  $x$ is the number of elements in cluster $k_i$ of $\mathcal{D}$;
  \item $y=x-m_{\mathcal D}$;
  \item $z=\mathcal{L}^{\mathcal{D}}(a_{\cal D})+\mathcal{L}^{\mathcal{D}}(b_{\cal D})$ if $m_{\mathcal{D}}=2$, and $z=\mathcal{L}^{\mathcal{D}}(a_{\cal D})-1$ if $m_{\mathcal{D}}=1$;
  \item $t$ is the total number of elements in all other non-singleton clusters of $\mathcal{D}$ except cluster $k_i$.
  \end{itemize}
%We agree that, when cluster $k_i$ becomes a singleton, the values of $y$ and $z$ in tuple $(x,y,z,t)$ are both $0$.

%we agree that, in the context of a given cirquent $\mathcal{D}$, every time Rule II is applied (bottom-up) to a disjunctive oconnective $c$, the level of $c$ decreases by 1; every time Rule III is applied (bottom-up) to two oconnectives $c$ and $d$, both $\mathcal{L}(c)$ and $\mathcal{L}(d)$ decrease by 1, with $\underline{cd}$ unchanged; and when cluster $k_i$ becomes a singleton one, the values of $y$ and $z$ in tuple $(x,y,z,t)$ are both 0.

Further, we define the relation ``$\leq$" on the set of all such tuples as follows. For any two tuples $(x_1,y_1,z_1,t_1)$ and $(x_2,y_2,z_2,t_2)$, $(x_1,y_1,z_1,t_1)\leq(x_2,y_2,z_2,t_2)$ if and only if one of the following conditions holds: (i) $x_1< x_2$; (ii) $x_1=x_2$ and $y_1< y_2$; (iii) $x_1=x_2$, $y_1=y_2$ and $z_1< z_2$; (iv) $x_1=x_2$, $y_1=y_2$, $z_1=z_2$ and $t_1< t_2$; (v) $x_1=x_2$, $y_1=y_2$, $z_1=z_2$ and $t_1=t_2$. It is easy to see that ``$\leq$" well-orders the set of all states, with each tuple (state) denoting an ordinal $<\omega^{4}$.

Now we show that every step of the process strictly decreases the state of the current cirquent. One can see that the state of the current cirquent depends on cluster $k_i$ and the picked pair of $a,b$. At the beginning of this procedure, for the picked cluster $k_i$ and the picked pair of $a,b$ in step 1, the state $(x,y,z,t)$ has its original value. In step 2, every time substeps (i),(ii) are performed for $a$ (resp. $b$), $x,y$ do not change, but $z$ decreases by 1; then in step 3, when Rule III is applied, $x$ is decreased by 1. As long as cluster $k_i$ is not a singleton, the process will come into the next iteration of the loop. During each iteration, a new pair of $a,b$ is picked and the value of $m$ is changed from $1$ to $2$ in step 1, which makes $x$ unchanged but $y$ is decreased by 1; then every iteration of the inner loop in step 2 leaves $x,y$ unchanged but decreases $z$ by 1; and then step 3 decreases $x$ by 1. Thus, the state keeps decreasing during the procedure, meaning that the latter terminates at some point.

Then after applying the three-step procedure, we get the resulting cirquent $\mathcal{C}_{2}$ where the number of non-singleton clusters is $n-1$, since no new non-singleton clusters are created during the procedure and cluster $k_i$ became a singleton. Pick any cluster $k_j$ ($j\neq i, 1\leq j\leq n$) in $\mathcal{C}_{2}$ and carry out the same steps as we did with $\mathcal{C}_{1}$, then we get the resulting cirquent $\mathcal{C}_{3}$, with $n-2$ non-singleton clusters. Every time the above steps are performed, the number of non-singleton clusters in the current (topmost) cirquent decreases by 1. Therefore, the outermost procedure will end sooner or later with the resulting cirquent $\mathcal{C}_{n+1}$ having no non-singleton clusters. Since only Rule II and Rule III are applied during this procedure when we get $\mathcal{C}_{n+1}$ from $\mathcal{C}_1$ and both of these rules preserve truth in the bottom-up direction, all the cirquents $\mathcal{C}_{2}, \mathcal{C}_{3}, \ldots, \mathcal{C}_{n+1}$ are valid.
Hence, our construction of a proof of $\mathcal{C}$ ends up with the top most cirquent $\mathcal{C}_{n+1}$, which is an axiom of $IF_p$.
\end{proof}

%\section{Further results}
%The features of this system:

%(i)It's a conservative extension of classical propositional logic.

%(ii)It accounts for dependence from disjunctions and hence is the first step for giving a formal system for propositional IF logic.

%(iii)Further works will be introducing a system that can account for dependence from disjunctions and conjunctions.

\end{document}